\documentclass[11pt,british,reqno]{article}
\usepackage{lmodern}
\usepackage{avant}
\usepackage{courier}

\usepackage[T1]{fontenc}
\usepackage[latin9]{inputenc}
\usepackage[a4paper]{geometry}
\usepackage{bbold}
\usepackage{stmaryrd}
\SetSymbolFont{stmry}{bold}{U}{stmry}{m}{n}
\usepackage{color}
\usepackage{setspace}
\usepackage[titletoc,title]{appendix}
\usepackage{float}
\usepackage{babel}
\usepackage{pifont}
\usepackage{amsmath}
\usepackage{amsthm}
\usepackage{amssymb}
\usepackage{latexsym}
\usepackage[all]{xy}
\usepackage{mathrsfs}
\usepackage{graphicx}
\usepackage{setspace}
\usepackage[numbers]{natbib}
\usepackage[unicode=true,pdfusetitle,
 bookmarks=true,bookmarksnumbered=false,bookmarksopen=false,
 breaklinks=false,pdfborder={0 0 0},pdfborderstyle={},backref=false,colorlinks=true]
 {hyperref}
\hypersetup{
 linkcolor=magenta, urlcolor=blue, citecolor=blue}

\makeatletter
\numberwithin{equation}{section}
\numberwithin{figure}{section}
\numberwithin{table}{section}
  \theoremstyle{plain}
  \newtheorem{prop}{\protect\propositionname}

\makeatletter
\renewenvironment{proof}[1][\proofname]{\par
  \pushQED{\qed}%
  \normalfont \topsep6\p@\@plus6\p@\relax
  \list{}{%
    \settowidth{\leftmargin}{\itshape\proofname:\hskip\labelsep}%
    \setlength{\labelwidth}{0pt}%
    \setlength{\itemindent}{-\leftmargin}%
  }%
  \item[\hskip\labelsep\itshape#1\@addpunct{:}]\ignorespaces
}{%
  \popQED\endlist\@endpefalse
}
\makeatother

\providecommand{\keywords}[1]{\textsc{\textit{Keywords---}} #1}
\newcommand{\idd}{\mathbb{1}}

\date{}

\providecommand{\propositionname}{Proposition}

\begin{document}

\title{Entropy driven transformations of statistical hypersurfaces}

\author{Mario Angelelli and Boris Konopelchenko}
\maketitle
\begin{center}
Department of Mathematics and Physics ``Ennio De Giorgi'', 
\par\end{center}

\begin{center}
University of Salento 
\par\end{center}

\begin{center}
and INFN, Sezione di Lecce
\par\end{center}

\begin{center}
Lecce 73100, Italy 
\par\end{center}
\begin{abstract}
Deformations of geometric characteristics of statistical hypersurfaces
governed by the law of growth of entropy are studied. Both general
and special cases of deformations are considered. The basic structure of the statistical hypersurface is explored through a differential relation for the variables, and connections with the replicator dynamics for Gibbs' weights are highlighted. Ideal and super-ideal cases are analysed, also considering their integral characteristics.
\end{abstract}


\section{\label{sec: Introduction} Introduction}

Surfaces, hypersurfaces and their dynamics (deformations) are key
ingredients in a broad variety of mathematical problems and phenomena
in physics (see e.g. \cite{Darboux1890,Eisenhart1909,Yau2016,Nelson2004,Gross1992,David1996}).
In mathematics, the mechanisms governing the deformation of surfaces
vary from the classical one, which preserves certain characteristics
of surfaces (see e.g. \cite{Darboux1890,Eisenhart1909,Yau2016}),
to those described by integrable partial differential equations \cite{Konopelchenko1996,Konopelchenko1998,Konopelchenko2000}.
In physics, surfaces and hypersurfaces are forced to deform due to
certain effects varying from a simple change of pressure inside a
soap bubble to the interaction of world-sheets with background in
string theory (see e.g. \cite{Nelson2004,Gross1992,David1996}). 

Recently we considered \cite{AK2016} a special class of hypersurfaces
inspired by the basic formula for the free energy in statistical physics
(see e.g. \cite{LL1980}). This approach is aimed at identifying and, then, relating the microscopic and macroscopic sectors in statistical models. Such a purpose suggests introducing more spaces and a mapping relating them. The focus on mappings between spaces allows one to deal with \emph{families} of statistical models, which is a useful approach in the analysis of complex systems. 

In \cite{AK2016} we focused on parametric families of systems. In particular, statistical hypersurfaces have been defined by the formula 
\begin{equation}
x_{n+1}=\ln\left(\sum_{\alpha=1}^{m}e^{f_{\alpha}(x_{1},\dots,x_{n})}\right)\label{eq: basic free energy formula}
\end{equation}
where $x_{1},\dots,x_{n},x_{n+1}$ are Cartesian coordinates in $\mathbb{R}^{n+1}$ and $f_{\alpha}(x_{1},\dots,x_{n})$ are certain functions. 
The mapping (\ref{eq: basic free energy formula}) leads to the representation of the set of parameters as a $n$-dimensional hypersurface embedded in a $(n+1)$-dimensional ambient space. If the ambient space is endowed with a metric structure, this hypersurface inherits metric properties, which can be used to study and compare characteristic features of statistical models defined by different parameters. 
The role of physical interactions and phase transitions \cite{LL1980} can be easily discussed in this geometric framework. In particular, the notion on ``ideal model'' can be related to a special (translational) symmetry. On the other hand, the basic geometric characteristics of such hypersurfaces have rather special properties, which can be expressed in terms of mean values calculated with Gibbs' distribution 
\begin{equation}
w_{\alpha}:=\frac{e^{f_{\alpha}}}{\sum_{\beta=1}^{m}e^{f_{\beta}}},\quad\alpha=1,\dots,m.\label{eq: Generalized Gibbs' weights}
\end{equation}
Within such probabilistic view, the entropy defined by the standard formula \cite{LL1980} 
\begin{equation}
S:=-\sum_{\alpha=1}^{m}w_{\alpha}\ln w_{\alpha}\label{eq: generalized entropy}
\end{equation}
has a simple geometrical meaning. Namely \cite{AK2016}, 
\begin{equation}
S=x_{n+1}-\overline{f}\label{eq: generalized entropy, geometric}
\end{equation}
where $\overline{f}$ is the mean value $\overline{f}:=\sum_{\alpha=1}^{m}w_{\alpha}f_{\alpha}$.
For the super-ideal case $m=n$ and $f_{\alpha}=x_{\alpha}$, the
entropy (\ref{eq: generalized entropy, geometric}) is $S=\sqrt{\det\mathbf{g}}\vec{X}\cdot\vec{N}$,
where $\vec{X}$ and $\vec{N}$ are the position vector and the normal
vector at the point of the hypersurface (\ref{eq: basic free energy formula}),
respectively. 

The natural appearance of entropy (\ref{eq: generalized entropy})-(\ref{eq: generalized entropy, geometric}) for the statistical hypersurfaces (\ref{eq: basic free energy formula})
suggests to consider the class of their deformations, transformations,
and dynamics which are governed by the most fundamental property of
entropy, namely, by the law of growth of entropy for isolated macroscopic
systems \cite{LL1980}. 

The present paper is devoted to the analysis of geometrical characteristics
of such entropy-driven transformations of statistical hypersurfaces.
It is shown that the general deformations of statistical hypersurfaces
corresponding to increasing entropy are characterised by the conditions
\[
\overline{\delta f}>\delta\overline{f},\quad\overline{f\cdot\delta f}<\overline{f}\cdot\delta\overline{f}.
\]
It is also highlighted that the variation of principal curvatures
is linked to the tendency to stability or to a departure from it. 

We extend the analysis in \cite{AK2016} discussing two classes of deformations. The first class includes deformations due to variations of variables $\boldsymbol{x}$ (parametric transformations), while the second class includes non-parametric transformations where functions $f_{\alpha}$ are changed. The case of affine functions
$f_{\alpha}$ and, in particular, linear and super-ideal cases are
studied in more details. It is shown that, in the linear case, the
variation of the entropy between two regions of different statistical
hypersurfaces is proportional to the volume of a domain in $\mathbb{R}^{n+1}$
bounded by them. A connection with replicator dynamics is also noted. 

The paper is organized as follows. Some general formulas for statistical
hypersurfaces are presented in Section \ref{sec: Statistical hypersurfaces}.
General variations of entropy and their properties are considered
in Section \ref{sec: General variations of entropy }. The interrelation
between principal curvatures and the second law of thermodynamics
is discussed in Section \ref{sec: Principal curvatures and entropy}. 
The effects of deformations on statistical weights are investigated in Section \ref{sec: Effects of deformations on statistical weights} in terms of a differential characterization of the statistical mapping and a relation with the replicator dynamics.
Ideal and super-ideal statistical hypersurfaces are studied in Section
\ref{sec: Ideal cases}. Some possible future developments are noted in Section \ref{sec: Conclusion}. 

\section{Statistical hypersurfaces and their deformations}
\label{sec: Statistical hypersurfaces}

\subsection{Brief summary on statistical hypersurfaces and their tropical limit}
\label{subsec: Brief summary on statistical hypersurfaces and their tropical limit}

Here, for convenience, we first reproduce some basic results of the
paper \cite{AK2016}. In the rest of the paper, we will denote the local coordinates in $\mathbb{R}^{n}$ and in $\mathbb{R}^{m}$ as $\boldsymbol{x}$ and $\vec{f}$, respectively.

The induced metric $\mathbf{g}$ of the hypersurface $V_{n}\subseteq\mathbb{R}^{n+1}$
defined by the formula (\ref{eq: basic free energy formula}) is 
\begin{equation}
g_{ik}=\delta_{ik}+\overline{f_{i}}\cdot\overline{f_{k}}\label{eq: metric hypersurface}
\end{equation}
where 
\begin{equation}
\overline{f_{i}}:=\sum_{\alpha=1}^{m}w_{\alpha}\frac{\partial f_{\alpha}}{\partial x_{i}},\quad i=1,\dots,n.\label{eq: averages derivatives}
\end{equation}
The position vector $\vec{X}$ for a point on $V_{n}$ and the corresponding
normal vector $\vec{N}$ are 
\begin{align}
\vec{X}= & (x_{1},\dots,x_{n},x_{n+1}),\nonumber \\
\vec{N}= & \frac{1}{\sqrt{\det\mathbf{g}}}\left(-\overline{f_{1}},\dots,-\overline{f_{n}},1\right)\label{eq: position and normal vectors}
\end{align}
where $\det\mathbf{g}=1+\sum_{i=1}^{n}\overline{f_{i}}^{2}$. 

The second fundamental form $\Omega_{ik}$ is given by 
\begin{equation}
\Omega_{ik}=\frac{1}{\sqrt{\det\mathbf{g}}}\left(\overline{f}_{(ik)}-\overline{f_{i}}\cdot\overline{f_{k}}\right)\label{eq: second fundamental form components}
\end{equation}
where 
\begin{equation}
\overline{f}_{(ik)}:=\sum_{\alpha=1}^{m}w_{\alpha}\left(\frac{\partial^{2}f_{\alpha}}{\partial x_{i}\partial x_{k}}+\frac{\partial f_{\alpha}}{\partial x_{i}}\cdot\frac{\partial f_{\alpha}}{x_{k}}\right),\quad i,k=1,\dots,n.\label{eq: correlator}
\end{equation}

The Riemann curvature tensor is 
\begin{equation}
R_{iklj}=\Omega_{il}\Omega_{kj}-\Omega_{kl}\Omega_{ij},\quad i,k,l,j=1,\dots,n.\label{eq: Riemann curvature tensor covariant components}
\end{equation}
and the Gauss-Kronecker curvature is given by 
\begin{equation}
K=\frac{\det\left|\overline{f}_{(ik)}-\overline{f_{i}}\cdot\overline{f_{k}}\right|}{\left(1+\sum_{l=1}^{n}\overline{f_{l}}^{2}\right)^{\frac{n+2}{2}}}.\label{eq: Gauss-Kronecker curvature}
\end{equation}
The entropy $S$ (\ref{eq: generalized entropy, geometric}) is equivalent
to \cite{AK2016} 
\begin{equation}
S=\sqrt{\det\mathbf{g}}\vec{X}\cdot\vec{N}+\sum_{\alpha=1}^{m}w_{\alpha}\left(\sum_{i=1}^{n}x_{i}\frac{\partial f_{\alpha}}{\partial x_{i}}-f_{\alpha}\right)\label{eq: entropy geometric formula}
\end{equation}
which takes the form 
\begin{equation}
S=\sqrt{\det\mathbf{g}}\vec{X}\cdot\vec{N}\label{eq: entropy geometric formula, super-ideal}
\end{equation}
in the super-ideal case ($f_{\alpha}=x_{\alpha}$). In all these formulae,
a hypersurface associated with the physical system is specified by
the choice of functions $f_{\alpha}(\boldsymbol{x})$, while $x_{1},\dots,x_{n}$
are local coordinates (parameters which characterise the state of
the system). 

The definition of a geometric framework to analyse statistical models also gives a natural setting to study their \emph{tropical limit}, which is related to Maslov's dequantization and ultradiscretization procedures. Tropical structures can be obtained from standard (real or complex) ones through the limit $\varepsilon\rightarrow0^{+}$ for the following family of ring operations
\begin{eqnarray}
x\oplus_{\varepsilon}y & := & \varepsilon\cdot\ln\left(e^{\frac{x}{\varepsilon}}+e^{\frac{x}{\varepsilon}}\right),\nonumber \\
x\odot_{\varepsilon}y & := &\varepsilon\cdot\ln\left(e^{\frac{x}{\varepsilon}}\cdot e^{\frac{x}{\varepsilon}}\right)=x+y.
\label{eq: tropical limit}
\end{eqnarray}

The exploration of the tropical limit in statistical physics naturally follows from the expression (\ref{eq: basic free energy formula}). In this case, the Boltzmann constant $k_{B}$ plays the role of the parameter $\varepsilon$ in the formal limit (\ref{eq: tropical limit}). A first study in this direction was carried out in \cite{AK2015} with $f_{\alpha}=\frac{T\cdot S_{\alpha}-E_{\alpha}}{k_{B}T}$, where $T$ is the absolute temperature, $\{E_{\alpha}\}$ is the energy spectrum of the system, and $\exp\frac{S_{\alpha}}{k_{B}}$ represents the degeneracy associated with $E_{\alpha}$. This approach offers new techniques to deal with systems with highly degenerated energy levels, e.g. frustrated systems. 

The tropical limit of statistical models has been extended to more general cases in \cite{AK2016} to meet the formalism of statistical hypersurfaces. This extension has led to the identification of different kinds of tropical limit based on the choice of independent variables $\boldsymbol{x}$ or dependent ones $f_{\alpha}$ for the limit procedure. 

An investigation of formal tropical structures in the context of statistical physics has been carried out in \cite{Angelelli2017} focusing on algebraic aspects.

\subsection{General deformations of statistical hypersurfaces}
\label{subsec: General deformations of statistical hypersurfaces}

Now let us consider deformations of a statistical hypersurface generated
by general infinitesimal variations of functions $f_{\alpha}\mapsto f_{\alpha}+\delta f_{\alpha}$,
$\alpha=1,\dots,m$. Using the expressions (\ref{eq: metric hypersurface})-(\ref{eq: Gauss-Kronecker curvature}),
one easily gets the general variations of these quantities: 
\begin{equation}
\delta g_{ik}=\delta\overline{f_{i}}\cdot\overline{f_{k}}+\overline{f_{i}}\cdot\delta\overline{f_{k}},\label{eq: variation metric}
\end{equation}
\begin{equation}
\delta\Omega_{ik}=-\frac{1}{2}\frac{\text{tr}(\delta\mathbf{g})}{(\det\mathbf{g})^{3/2}}\cdot\left(\overline{f}_{(ik)}-\overline{f_{i}}\cdot\overline{f_{k}}\right)+\frac{1}{(\det\mathbf{g})^{1/2}}\left(\delta\overline{f}_{(ik)}-\delta g_{ik}\right)\label{eq: variation second fundamental form}
\end{equation}
and 

\begin{equation}
\delta K=\frac{\delta(\det\left|\overline{f}_{(ik)}-\overline{f_{i}}\cdot\overline{f_{k}}\right|)}{\det\mathbf{g}^{(n+2)/2}}-(n+2)\cdot\frac{\det\left|\overline{f}_{(ik)}-\overline{f_{i}}\cdot\overline{f_{k}}\right|\cdot\text{tr}(\delta\mathbf{g})}{\det\mathbf{g}^{(n+4)/2}}.\label{eq: variation Gauss-Kronecker curvature}
\end{equation}
In cases where the Hessian matrix of $F$ 
\begin{equation}
\boldsymbol{\partial^{2}}F:=\left\llbracket \overline{f}_{(ik)}-\overline{f_{i}}\cdot\overline{f_{k}}\right\rrbracket _{i,k}\label{eq: Hessian of F}
\end{equation}
is invertible, one can use the well-known identity for the variation
of matrices 
\[
\det\left(\boldsymbol{\partial^{2}}F+\delta\boldsymbol{\partial^{2}}F\right)=\det\left(\boldsymbol{\partial^{2}}F\right)\cdot\text{tr}\left(\boldsymbol{\partial^{2}}F^{-1}\cdot\delta\boldsymbol{\partial^{2}}F\right)
\]
and get 
\begin{align}
\delta K= & \frac{\text{tr}\left(\left\llbracket \overline{f}_{(ik)}-\overline{f_{i}}\cdot\overline{f_{k}}\right\rrbracket _{i,k}^{-1}\cdot\left\llbracket \delta\overline{f}_{(ik)}-\delta\overline{f_{i}}\cdot\overline{f_{k}}-\overline{f_{i}}\cdot\delta\overline{f_{k}}\right\rrbracket _{i,k}\right)}{\det\left(\mathbf{g}^{(n+2)/2}\cdot\left\llbracket \overline{f}_{(ik)}-\overline{f_{i}}\cdot\overline{f_{k}}\right\rrbracket _{i,k}^{-1}\right)}\nonumber \\
 & -(n+2)\cdot\frac{\det\left\llbracket \overline{f}_{(ik)}-\overline{f_{i}}\cdot\overline{f_{k}}\right\rrbracket _{i,k}\cdot\text{tr}(\delta\mathbf{g})}{\det\mathbf{g}^{(n+4)/2}}.\label{eq: variation Gauss-Kronecker curvature, H!=00003D0}
\end{align}
Note that, if $K=0$ and some conditions on the original functions
are assumed, then the first variation of $K$ vanishes too. For variations
of the Riemann tensor we have 
\begin{align}
\delta R_{iklj}= & \frac{\text{tr}(\delta\mathbf{g})}{(\det\mathbf{g})^{2}}\cdot\left\{ \left(\overline{f}_{(kl)}-\overline{f_{k}}\cdot\overline{f_{l}}\right)\left(\overline{f}_{(ij)}-\overline{f_{i}}\cdot\overline{f_{j}}\right)-\left(\overline{f}_{(il)}-\overline{f_{i}}\cdot\overline{f_{l}}\right)\left(\overline{f}_{(kj)}-\overline{f_{k}}\cdot\overline{f_{j}}\right)\right\} \nonumber \\
+ & \frac{(\delta\overline{f}_{(il)}-\delta g_{il})\cdot(\overline{f}_{(kj)}-\overline{f_{k}}\cdot\overline{f_{j}})+(\delta\overline{f}_{(kj)}-\delta g_{kj})\cdot(\overline{f}_{(il)}-\overline{f_{i}}\cdot\overline{f_{l}})}{\det\mathbf{g}}.\nonumber \\
- & \frac{(\delta\overline{f}_{(kl)}-\delta g_{kl})\cdot(\overline{f}_{(ij)}-\overline{f_{i}}\cdot\overline{f_{j}})+(\delta\overline{f}_{(ij)}-\delta g_{ij})\cdot(\overline{f}_{(kl)}-\overline{f_{k}}\cdot\overline{f_{l}})}{\det\mathbf{g}}.\label{eq: variation Riemann tensor}
\end{align}
while the scalar curvature \cite{AK2016} varies as 
\begin{align}
\delta R= & 2\cdot\text{tr}\boldsymbol{\Omega}\cdot\text{tr}(\delta\boldsymbol{\Omega})-\text{tr}(\delta\boldsymbol{\Omega}^{2})+2\frac{\delta((\boldsymbol{\overline{f}})^{T}\cdot\boldsymbol{\Omega}^{2}\cdot\boldsymbol{\overline{f}})\cdot(1-\text{tr}(\boldsymbol{\Omega}))-\text{tr}(\delta\boldsymbol{\Omega})\cdot((\boldsymbol{\overline{f}})^{T}\cdot\boldsymbol{\Omega}\cdot\boldsymbol{\overline{f}})}{\det\mathbf{g}}\nonumber \\
 & -2\frac{(\boldsymbol{\overline{f}})^{T}\cdot\boldsymbol{\Omega}^{2}\cdot\boldsymbol{\overline{f}}-\text{tr}(\boldsymbol{\Omega})\cdot((\boldsymbol{\overline{f}})^{T}\cdot\boldsymbol{\Omega}\cdot\boldsymbol{\overline{f}})}{\det\mathbf{g}^{2}}\cdot(\boldsymbol{\overline{f}}\cdot\delta\boldsymbol{\overline{f}}).\label{eq: variation scalar curvature}
\end{align}

\section{\label{sec: General variations of entropy } General variations of
entropy }

Now we will consider variations of entropy generated by general infinitesimal
variations of the functions $f_{\alpha}\mapsto f_{\alpha}+\delta f_{\alpha}$,
$\alpha=1,\dots m$. The formula (\ref{eq: generalized entropy, geometric})
implies that 
\begin{equation}
\delta S=\overline{\delta f}-\delta\overline{f}\label{eq: variaton entropy intertwining}
\end{equation}
where $\overline{\delta f}:=\sum_{\alpha=1}^{m}w_{\alpha}\delta f_{\alpha}$. 

Variations of $f_{\alpha}$ generate the variations of probabilities
$w_{\alpha}$, namely, 
\begin{align}
\delta w_{\alpha}= & \sum_{\beta=1}^{m}(\delta_{\alpha\beta}w_{\alpha}-w_{\alpha}w_{\beta})\delta f_{\beta}\nonumber \\
= & w_{\alpha}\cdot(\delta f_{\alpha}-\overline{\delta f}).\label{eq: variation Gibbs' weights}
\end{align}
Since $\delta S=-\sum_{\alpha=1}^{m}\delta w_{\alpha}\cdot f_{\alpha}$,
using (\ref{eq: variation Gibbs' weights}) one also gets 
\begin{align}
\delta S= & -\sum_{\alpha=1}^{m}w_{\alpha}f_{\alpha}\delta f_{\alpha}+\left(\sum_{\alpha=1}^{m}w_{\alpha}f_{\alpha}\right)\cdot\left(\sum_{\beta=1}^{m}w_{\beta}\delta f_{\beta}\right)\nonumber \\
= & -\frac{1}{2}\cdot\overline{\delta(f{}^{2})}+\overline{f}\cdot\overline{\delta f}.\label{eq: entropy variation and fluctuations}
\end{align}
Thus, deformations of statistical hypersurfaces driven by the law
of growth of entropy ($\delta S>0$) are characterized by the conditions
\begin{align}
\overline{\delta f}> & \delta\overline{f},\label{eq: entropy variation direction, deg 1}\\
\overline{f\cdot\delta f}< & \overline{f}\cdot\delta\overline{f}.\label{eq: entropy variation direction, deg 2}
\end{align}
These inequalities select ``thermodynamic'' deformations of statistical
hypersurfaces among all possible. 

In the limiting case $\delta S=0$, one gets 
\begin{align}
\overline{\delta f}= & \delta\overline{f},\label{eq: entropy variation reversible, deg 1}\\
\overline{f\cdot\delta f}= & \overline{f}\cdot\delta\overline{f}.\label{eq: entropy variation reversible, deg 2}
\end{align}
Deformations defined by the conditions (\ref{eq: entropy variation reversible, deg 1})-(\ref{eq: entropy variation reversible, deg 2})
correspond to the equilibrium states in statistical systems or reversible
processes \cite{LL1980}. 

The variation of the entropy (\ref{eq: entropy variation and fluctuations})
can also be deduced by the covariance of two random variables, namely,
the variable $\vec{f}\in\mathbb{R}^{m}$ which assumes the value $f_{\alpha}$
with probability $w_{\alpha}$, and the associated variation that
occurs with the same probability $p(\delta f=\delta f_{\alpha})=w_{\alpha}$.
The constraint $\delta S\geq0$ on the sign of the variation of entropy
(\ref{eq: entropy variation and fluctuations}) means that these two
random variables have to be negatively correlated. Moreover, the condition
(\ref{eq: entropy variation reversible, deg 2}) means that, in the
case $\delta S=0$, the distributions of $f_{\alpha}$ and $\delta f_{\alpha}$
are statistically uncorrelated. Note also that, due to (\ref{eq: entropy variation reversible, deg 1}),
the condition (\ref{eq: entropy variation reversible, deg 2}) is
equivalent to 
\begin{equation}
\overline{\delta(f^{2})}=\delta(\overline{f}^{2}).\label{eq: quadratic, reversible}
\end{equation}
In order to demonstrate that the variables $\vec{f}$ and $\delta\vec{f}$
are not independent in general, let us consider variations $\vec{\delta f}$
that are not proportional to $(1,\dots,1)$, which is the vector associated
with a shift of the ground energy in (\ref{eq: basic free energy formula}).
We say that two random variables $\mathbf{a}$ and $\mathbf{b}$ supported
on the same finite set $\{1,\dots,m\}$ and with joint probability
$w(\mathbf{a}=a_{\alpha},\mathbf{b}=b_{\beta})=w_{\alpha\beta}$ are
\textit{totally uncorrelated} if the expressions 
\begin{equation}
\sum_{\alpha=1}^{m}\sum_{\beta=1}^{m}w_{\alpha\beta}a_{\alpha}^{u}b_{\beta}^{v}-\left(\sum_{\alpha=1}^{m}\sum_{\beta=1}^{m}w_{\alpha\beta}a_{\alpha}^{u}\right)\cdot\left(\sum_{\alpha=1}^{m}\sum_{\beta=1}^{m}w_{\alpha\beta}b_{\beta}^{v}\right)\label{eq: moments}
\end{equation}
vanish for all natural numbers $u,v$. We focus on exponential sums
(\ref{eq: basic free energy formula}) that return a given expression
with the minimal number of exponential terms, so that all the functions
$f_{\alpha}$ are pairwise distinct. One can always reduce to this
case since, if $f_{\alpha}=f_{\beta}$, $\alpha\neq\beta$, the pair
$(f_{\alpha},f_{\beta})$ can be substituted by $f_{\alpha}+\ln2$
without affecting the sum (\ref{eq: basic free energy formula}). 

Considering the variation for such a system, we get the following 
\begin{prop}
\label{prop: totally uncorrelated}The vectors $\mathbf{f}$ and $\delta\mathbf{f}$,
seen as random variables with joint probability $w(f_{\alpha},\delta f_{\beta})=\delta_{\alpha\beta}w_{\alpha}$,
are totally uncorrelated if and only if $\delta\mathbf{f}$ is proportional
to $(1,\dots,1)$. 
\end{prop}
\begin{proof}
Introducing the vectors 
\begin{equation}
\vec{f}_{(u)}:=(f_{1}^{u},\dots,f_{m}^{u}),\quad u\in\mathbb{N}\label{eq: powered exponents}
\end{equation}
we find that 
\begin{equation}
\overline{f^{u}\cdot\delta f}-\overline{(f)^{u}}\cdot\overline{\delta f}=\vec{f}_{(u)}^{T}\cdot\mathbf{H}\cdot\delta\vec{f}\label{eq: covariance with powers and bilinear}
\end{equation}
where the matrix $\mathbf{H}$ has entries 
\begin{equation}
(\mathbf{H})_{\alpha\beta}:=\delta_{\alpha\beta}w_{\alpha}-w_{\alpha}w_{\beta}.\label{eq: Hessian super-ideal skeleton}
\end{equation}
For a generic point $\boldsymbol{x}$, the Vandermonde determinant
\begin{equation}
\det(\vec{f}_{(1)},\vec{f}_{(2)},\dots,\vec{f}_{(m)})=|f_{\alpha}^{\beta-1}|_{\alpha,\beta=1}^{m}\cdot\prod_{\alpha=1}^{m}f_{\alpha}\label{eq: Vandermonde for dependence}
\end{equation}
is non-vanishing, since we are assuming that all the functions $f_{\alpha}$
are pairwise distinct. Thus, the vectors $\vec{f}_{(u)}$, $u=1,\dots,m$,
are linearly independent. If $\overline{f^{u}\cdot\delta f}-\overline{(f)^{u}}\cdot\overline{\delta f}$
vanishes, then (\ref{eq: covariance with powers and bilinear}) implies
that $\vec{f}_{(u)}$ lies in the orthogonal space to $\mathbf{H}\vec{\delta f}$.
Using the assumption that $\vec{\delta f}$ is not proportional to
$(1,\dots,1)$ and Proposition 3.1 in \cite{AK2016}, one easily
shows that this space has dimension $m-1$. But (\ref{eq: Vandermonde for dependence})
is non-vanishing, hence at least one of the correlations $\overline{f^{u}\cdot\delta f}-\overline{(f)^{u}}\cdot\overline{\delta f}$,
$u=1,\dots,m$, does not vanish. So $\vec{f}$ and $\delta\vec{f}$
are not totally uncorrelated. 

If instead $\vec{\delta f}$ is proportional to $(1,\dots,1)$, then
$\vec{\delta f}$ is the eigenvector of $\mathbf{H}$ corresponding
to the eigenvalue $0$, so (\ref{eq: covariance with powers and bilinear})
identically vanishes independently on $u$. The same holds for any
choice of $v$ in (\ref{eq: moments}). 
\end{proof}
We stress that one can write (\ref{eq: entropy variation and fluctuations})
as $\delta S=-\vec{f}^{T}\mathbf{H}\vec{\delta f}$. So the previous
results imply that the condition of stationary entropy $\delta S=0$
with respect to the variation $\vec{\delta f}$ can be refined by
higher-order conditions related to the correlations (\ref{eq: moments}),
and it is enough to satisfy $m$ such conditions to recover the only
trivial variation, that is the shift of the ground energy $\delta\vec{f}=c\cdot(1,\dots,1)$,
$c\in\mathbb{R}$. The role of the choice of the ground energy has
been highlighted in the algebraic study of the tropical limit in statistical
physics \cite{Angelelli2017}. On the other hand, in the cases where
not all the components of $f$ are pairwise equal, there exists a
component, say $f_{m}$, such that $f_{m}\neq\overline{f}$. Since
the condition (\ref{eq: entropy variation and fluctuations}) is linear
in the variations $\delta f$, we can solve (\ref{eq: entropy variation and fluctuations})
with respect to $\delta f_{m}$ and get 
\begin{equation}
\delta f_{m}=w_{m}^{-1}\left(f_{m}-\sum_{\alpha=1}^{m}w_{\alpha}f_{\alpha}\right)^{-1}\cdot\left[\left(\sum_{\alpha=1}^{m}w_{\alpha}f_{\alpha}\right)\cdot\left(\sum_{\beta=1}^{m-1}w_{\beta}\delta f_{\beta}\right)-\sum_{\alpha=1}^{m-1}w_{\alpha}f_{\alpha}\delta f_{\alpha}\right].\label{eq: variation for reversible, general case}
\end{equation}

Finally, the point out that condition $\delta S=0$ expressed by (\ref{eq: quadratic, reversible})
is equivalent to the equation 
\begin{equation}
\overline{f}\cdot\overline{\delta f}=\frac{1}{2}\cdot\overline{\delta(f{}^{2})}.\label{eq: virial-like equation}
\end{equation}
It is noted an analogy with the virial theorem \cite{LL1980}, since
both cases describe a balancing condition between two contributions.
In our case, the balancing condition (\ref{eq: virial-like equation})
defines an equilibrium situation, namely, a reversibility condition. 

\section{\label{sec: Principal curvatures and entropy} Principal curvatures
and the second law }

The Gauss-Kronecker curvature was recognised as a suitable quantity
to distinguish particular systems in \cite{AK2016}, since it identifies
some natural generalizations of ideal models in statistical physics.
This aspect can be related to the maximum entropy principle as follows. 

From the second fundamental form \cite{Eisenhart1909} 
\begin{equation}
\Omega_{ik}=\vec{N}\cdot\frac{\partial^{2}\vec{X}}{\partial x_{i}\partial x_{k}}=\frac{1}{\sqrt{\det\mathbf{g}}}\left(\overline{f}_{\{ik\}}-\overline{f}_{i}\cdot\overline{f}_{k}\right),\quad i,k=1,\dots,n\label{eq: second fundamental form}
\end{equation}
where \cite{AK2016} 
\begin{equation}
\overline{f}_{\{ik\}}:=\sum_{\alpha=1}^{m}w_{\alpha}\left(\frac{\partial^{2}f_{\alpha}}{\partial x_{i}\partial x_{k}}+\frac{\partial f_{\alpha}}{\partial x_{i}}\cdot\frac{\partial f_{\alpha}}{\partial x_{k}}\right),\quad i,k=1,\dots,n.\label{eq: second average}
\end{equation}
we get the Weingarten map (or shape operator)
\begin{align}
W_{\,j}^{i}= & \sum_{k=1}^{n}g^{ik}\Omega_{kj}\nonumber \\
= & \frac{1}{\sqrt{\det\mathbf{g}}}\left(\sum_{\alpha=1}^{m}w_{\alpha}\left(\frac{\partial^{2}f_{\alpha}}{\partial x_{i}\partial x_{j}}+\frac{\partial f_{\alpha}}{\partial x_{i}}\cdot\frac{\partial f_{\alpha}}{\partial x_{j}}\right)-\overline{f}_{i}\cdot\overline{f}_{j}\right)\nonumber \\
 & -\frac{\overline{f}_{i}}{(1+||\overline{f}||^{2})^{3/2}}\left(\sum_{k=1}^{n}\sum_{\alpha=1}^{m}w_{\alpha}\left(\overline{f}_{k}\frac{\partial^{2}f_{\alpha}}{\partial x_{j}\partial x_{k}}+\frac{\partial f_{\alpha}}{\partial x_{j}}\cdot\overline{f}_{k}\frac{\partial f_{\alpha}}{\partial x_{k}}\right)-\overline{f}_{j}\cdot(\overline{f}_{k})^{2}\right).\label{eq: Weingarten map, components}
\end{align}
The quantities $W_{\,j}^{i}$, $i,j\in\{1,\dots,n\}$, can be expressed
by 
\begin{equation}
\mathbf{W}=\det\mathbf{g}^{-3/2}\cdot\left((\det\mathbf{g})\cdot\idd-\nabla F\cdot\nabla F^{T}\right)\cdot\boldsymbol{\partial^{2}}F\label{eq: Weingarten map}
\end{equation}
where $\boldsymbol{\partial^{2}}F$ is the Hessian matrix of $F$
as defined in (\ref{eq: Hessian of F}). Thus, we can look at the
eigenvalues of $\mathbf{W}$, i.e. the principal curvatures. From
the matrix determinant lemma \cite{Gantmacher1977}, we get 
\begin{align}
\det\left((\det\mathbf{g}\cdot\idd)-\nabla F\cdot\nabla F^{T}\right)= & (\det\mathbf{g})^{n}\cdot\left(1-(\det\mathbf{g})^{-1}\cdot||\nabla F||^{2}\right)\nonumber \\
= & (\det\mathbf{g})^{n}\cdot\frac{\det\mathbf{g}-||\nabla F||^{2}}{\det\mathbf{g}}=(\det\mathbf{g})^{n-1}.\label{eq: determinant rank-1 perturbation}
\end{align}
From (\ref{eq: Weingarten map}) and (\ref{eq: determinant rank-1 perturbation})
one can see that $\det\mathbf{W}=0$ implies that the Hessian matrix
is singular, $\det(\boldsymbol{\partial^{2}}F)=0$. Thus the occurrence
of vanishing principal curvatures, which is equivalent to $\det(\mathbf{W})=0$,
is related to the non-invertibility of $\boldsymbol{\partial^{2}}F$. 

In the particular case of affine functions $f_{\alpha}(\boldsymbol{x})=b_{\alpha}+\sum_{i=1}^{n}a_{\alpha i}x_{i}$,
the terms $\frac{\partial^{2}f_{\alpha}}{\partial x_{i}\partial x_{k}}$
in (\ref{eq: second average}) vanish, so (\ref{eq: Weingarten map})
becomes 
\begin{align}
W_{\,j}^{i}= & \frac{1}{\sqrt{\det\mathbf{g}}}\left(\sum_{\alpha=1}^{m}w_{\alpha}\frac{\partial f_{\alpha}}{\partial x_{i}}\cdot\frac{\partial f_{\alpha}}{\partial x_{j}}-\overline{f}_{i}\cdot\overline{f}_{j}\right)\nonumber \\ 
& -\sum_{k=1}^{n}\frac{\overline{f}_{i}\overline{f}_{k}}{(1+||\overline{f}||^{2})^{3/2}}\left(\sum_{\alpha=1}^{N}w_{\alpha}\frac{\partial f_{\alpha}}{\partial x_{j}}\frac{\partial f_{\alpha}}{\partial x_{k}}-\overline{f}_{j}\cdot\overline{f}_{k}\right)\label{eq: shape operator, components affine case}
\end{align}
that is 
\begin{align}
\mathbf{W}= & \frac{1}{\det\mathbf{g}^{1/2}}\cdot\mathrm{cov}(\mathbf{\partial f},\mathbf{\partial f})-\frac{1}{\det\mathbf{g}^{3/2}}\cdot\nabla F\cdot\nabla F^{T}\cdot\mathrm{cov}(\mathbf{\partial f},\mathbf{\partial f})\nonumber \\
= & \det\mathbf{g}^{-3/2}\cdot\left((\det\mathbf{g}\cdot\idd)-\nabla F\cdot\nabla F^{T}\right)\cdot\mathrm{cov}(\mathbf{\partial f},\mathbf{\partial f})\label{eq: shape operator}
\end{align}
where we denote by $(\mathbf{\partial f})_{i}$ the random variable taking
values $\{\partial_{x_{i}}f_{1},\dots,\partial_{x_{i}}f_{m}\}$, $i\in\{1,\dots,n\}$, and we assume the joint probability $w\left(\partial_{x_{1}}f_{\alpha_{1}},\dots,\partial_{x_{n}}f_{\alpha_{n}}\right)=w_{\alpha_{1}}\cdot\delta_{\alpha_{1}\alpha_{2}}\cdot\dots\cdot\delta_{\alpha_{1}\alpha_{n}}$. So the Hessian matrix $\boldsymbol{\partial^{2}}F$ becomes a covariance
matrix $\mathrm{cov}(\mathbf{\partial f},\mathbf{\partial f})$ relative
to $\mathbf{\partial f}$. 

Therefore, the Hessian $\boldsymbol{\partial^{2}}F$ establishes a
connection between a geometric condition on the Gauss-Kronecker curvature
(more generally, on the principal curvatures of the statistical hypersurface)
and the second law of the thermodynamics, with special regard to the
stability of the statistical system. In fact, it has been noticed
that the convexity of the free energy and the concavity of the entropy
with respect to some proper variables are not only useful technical
requirements, but they also have fundamental physical implications
in relation to the second law of thermodynamics (see e.g. \cite{terHorst1987,Leff1996}
and references therein). In particular, the focus of \cite{terHorst1987}
is on internal energy. On the other hand, starting from (\ref{eq: basic free energy formula})
we can combine Gibbs' formulation with the study of convexity, as
it naturally emerges from the extrinsic characteristics of hypersurfaces
embedded in a metric space. Furthermore, this approach also draws
attention to systems which do \emph{not} satisfy standard convexity
assumptions, such as those with negative temperatures or metastable
states. Suitable generalisations of the basic statistical mapping
are required to deal with these phenomena \cite{Angelelli2017,AK2018}. 

In this way, the evolution of the principal curvatures of the statistical
hypersurface is linked to the tendency to stability, or a departure
from it. The occurrence of nonlinearities described by terms $\frac{\partial^{2}f_{\alpha}}{\partial x_{i}\partial x_{k}}$
in (\ref{eq: Weingarten map, components}) may result in particular
physical behaviours, which are typical of non-homogeneous systems.
For instance, these contributions may follow from the particular dependence
$g=g(\varepsilon)$ of the degeneration $g$ on the associated energy
level $\varepsilon$. Concrete examples in this regard involve large
systems with long-range interactions \cite{Gross2006,Bouchet2010},
or small systems with bounded spectrum \cite{Ramsey1956} and boundary
effects \cite{Buonsante2016}. Some of these features have been discussed
in a tropical setting \cite{Angelelli2017}, and we postpone a detailed
analysis of their geometric characteristics to a separate work.

\section{Effects of deformations on statistical weights}
\label{sec: Effects of deformations on statistical weights}

It is often assumed that the set of energy configurations and the
associated degenerations are known, so one recovers the equilibrium
distribution from the maximum entropy condition \cite{LL1980}, where
the weights $w_{\alpha}$ are functionals of $\vec{f}$. In our approach
we leave the functions $\vec{f}$ free in order to consider their
deformations $\delta\vec{f}$. However, we can still recover the standard
form for the free energy starting from the dependence of some ``generalised
weights'' $h_{\alpha}$ on the variables $f_{\beta}$, $\alpha,\beta\in\{1,\dots,m\}$. 

At this purpose, we assume the relation (\ref{eq: variation Gibbs' weights})
as a starting point to explore the effects of deformations $\delta f_{\alpha}$
on the weights and their relation with the stationary entropy condition. 

\subsection{\label{subsec: Differential characterization of the fundamental statistical mapping }
Differential characterization of the fundamental statistical mapping }

We look at functions $h_{\alpha}$, $\alpha=1,\dots,m$, which depend
on $\{f_{1},\dots,f_{m}\}$ through the following relation based on
(\ref{eq: variation Gibbs' weights})
\begin{equation}
{\displaystyle {\displaystyle \frac{\partial h_{\alpha}}{\partial f_{\beta}}=\delta_{\alpha\beta}h_{\beta}-h_{\alpha}h_{\beta}}},\quad\alpha,\beta\in\{1,\dots,m\}.\label{eq: ansatz for metric, 2}
\end{equation}
The requirement (\ref{eq: ansatz for metric, 2}) implies that $\sum_{\alpha=1}^{m}h_{\alpha}df_{\alpha}$
is closed. Hence it is exact on contractible domains, so there exists
a potential $\tilde{F}$ such that 
\begin{equation}
h_{\alpha}={\displaystyle \frac{\partial\tilde{F}}{\partial f_{\alpha}}}.\label{eq: potential for weights}
\end{equation}
Solving the equations (\ref{eq: ansatz for metric, 2}), one gets 
\begin{prop}
\label{prop: fundamental statistical map from balancing assumption}
Equation (\ref{eq: ansatz for metric, 2}) implies the existence of
a potential function $\tilde{F}$ such that $h_{i}=\partial_{f_{i}}\tilde{F}$,
where 
\begin{equation}
{\displaystyle \tilde{F}(\boldsymbol{x})=\log\left(\gamma+\sum_{\alpha=1}^{m}e^{f_{\alpha}(\boldsymbol{x})-\sigma_{\alpha}}\right)}.\label{eq: fundamental statistical map from balancing assumption}
\end{equation}
\end{prop}
\begin{proof}
See Appendix \ref{sec: Appendix A}. 
\end{proof}
The coefficient $\gamma$ in (\ref{eq: fundamental statistical map from balancing assumption})
is an index for the loss of translational invariance of energies by
the same quantity, and it is related to the possibility
to get a normalized probability distribution. Indeed, this aspect has been discussed through the notion of local tropical symmetry in \cite{Angelelli2017}, where tropical copies introduce a structural term that is described by $\gamma$ in the present framework. The possibility to recover a ``standard'' probability distribution $\sum_{\alpha=1}^{m}h_{\alpha}=1$ entails the condition $\gamma=0$ in (\ref{eq: fundamental statistical map from balancing assumption}). So the term $\gamma$ describes
a structural parameter that is not allowed to vary, namely, it is
not involved in the system (\ref{eq: ansatz for metric, 2}) of differential
equations. This distinction between the structural parameter $\gamma$
and varying quantities $\vec{f}$ is also reflected in those variations
induced by independent variables $\boldsymbol{x}$. 

It is worth commenting briefly on a graph-theoretic interpretation
of the assumption (\ref{eq: ansatz for metric, 2}). Let us consider
a complete graph with vertices $\mathcal{V}=\{1,\dots,m\}$ including
loops, that is edges of the form $(\alpha,\alpha)$. From the assignment
of weights $w_{(\alpha\beta)}$ for edges, $\alpha,\beta=1,\dots,m$,
we can consider the degree matrix $\mathbf{D}:=\mathrm{diag}(D_{1},\dots,D_{m})$,
where 
\begin{equation}
D_{\alpha}:=\sum_{\beta=1}^{m}w_{(\alpha\beta)},\label{eq: degree matrix for graph}
\end{equation}
the adjacency matrix $\mathbf{W}:=\left\llbracket w_{(\alpha\beta)}\right\rrbracket _{\alpha,\beta=1,\dots,m}$
and, hence, the Laplacian 
\begin{equation}
\mathbf{L}:=\mathbf{D}-\mathbf{W}\label{eq: Laplacian matrix for graph}
\end{equation}
whose components of $\boldsymbol{L}$ are 
\begin{equation}
L_{\alpha\beta}:=\begin{cases}
-w_{(\alpha\beta)} & \alpha\neq\beta,\\
D_{\alpha}-w_{(\alpha\alpha)} & \alpha=\beta.
\end{cases}\label{eq: constrained Laplacian matrix for graph}
\end{equation}
For any assignment of weights $w_{\alpha}$ to each \emph{node} $\alpha$,
the consistency constraint between the weights for nodes and for edges
imposes that 
\begin{equation}
D_{\alpha}=\sum_{\beta=1}^{m}w_{(\alpha\beta)}=w_{\alpha},\quad\alpha=1,\dots,m.\label{eq: balancing constraints for each vertex}
\end{equation}
If $\sum_{\alpha=1}^{m}w_{\alpha}=1$, one can look at $w_{(\alpha\beta)}$
as the joint probability of random variables over $\{1,\dots,m\}\times\{1,\dots,m\}$
and the condition (\ref{eq: balancing constraints for each vertex})
as the definition of marginal probabilities. Then, the requirement
(\ref{eq: balancing constraints for each vertex}) is a balancing
condition, or local conservation law, between the node and edge weights.
Specifically, this means that the M\"{o}bius combination 
\begin{equation}
w_{\alpha}-\sum_{\beta=1}^{m}w_{(\alpha\beta)}\label{eq: Mobius combination}
\end{equation}
vanishes for each individual $\alpha=1,\dots,m$. If we choose the
joint distribution for independent variables $w_{(\alpha\beta)}=w_{\alpha}w_{\beta}$
in (\ref{eq: constrained Laplacian matrix for graph}), we recover
the right-hand side of (\ref{eq: ansatz for metric, 2}). Thus, the
evolution of functions $h_{\alpha}$ with respect to varying quantities
$f_{\beta}$, $\alpha,\beta\in\{1,\dots,m\}$, is described by the
Laplacian of the associated weighted graph (\ref{eq: constrained Laplacian matrix for graph}). 

\subsection{\label{subsec: Comments on replicator dynamics } Comments on replicator
dynamics }

The variation of entropy (\ref{eq: entropy variation and fluctuations})
can now be discussed in terms of its effects on Gibbs' distribution.
At this purpose, we start from the assumption that the function $f_{\alpha}$
are bounded in each neighborhood of a point $\boldsymbol{x}.$ If
this condition does not hold, the associated singularities are interpreted
as phase transitions \cite{AK2016}. Thus, for all $\boldsymbol{y}$
in such a neighborhood of $\boldsymbol{x}$ the expressions for Gibbs'
weights are preserved under the translation $f_{\alpha}\mapsto f_{\alpha}+M$,
where $M$ can be chosen in order to satisfy 
\begin{equation}
M>-{\displaystyle \min_{\alpha}}\left\{ f_{\alpha}(\boldsymbol{y})\right\} .\label{eq: shift ground energy}
\end{equation}
This invariance under the shift $f_{\alpha}\mapsto f_{\alpha}+M$
becomes non-trivial when different tropical algebraic structures
are taken into account. In order to compare different choices for the ground energy, the action of the shift on $f_{\alpha}$ or $\delta f_{\alpha}$ can be explored in terms of local tropical symmetry \cite{Angelelli2017}. 

Then, a new probability distribution $w^{(1)}$ can be introduced
\begin{equation}
w_{\alpha}^{(1)}(\boldsymbol{x}):=\frac{e^{f_{\alpha}}\cdot f_{\alpha}}{\sum_{\beta=1}^{m}e^{f_{\beta}}f_{\beta}},\quad\alpha\in\{1,\dots,m\}\label{eq: discrete replicator dynamics, f, first step}
\end{equation}
which is the first element of the sequence 
\begin{equation}
w_{\alpha}^{(t+1)}=w_{\alpha}^{(t)}\cdot\frac{f_{\alpha}^{(t)}}{\sum_{\beta=1}^{m}f_{\beta}^{(t)}w_{\beta}^{(t)}}.\label{eq: discrete replicator dynamics}
\end{equation}
The evolution of this sequence of probability distributions corresponds
to a kind of discrete replicator dynamics \cite{Bomze1983,Diederich1989}.
In fact, we can express (\ref{eq: discrete replicator dynamics})
as 
\begin{equation}
w_{\alpha}^{(t+1)}-w_{\alpha}^{(t)}=w_{\alpha}^{(t)}\cdot\left(\frac{f_{\alpha}^{(t)}}{\sum_{\beta=1}^{m}f_{\beta}^{(t)}w_{\beta}^{(t)}}-\sum_{\gamma=1}^{m}\frac{f_{\gamma}^{(t)}}{\sum_{\beta=1}^{m}f_{\beta}^{(t)}w_{\beta}^{(t)}}w_{\gamma}^{(t)}\right),\label{eq: discrete replicator dynamics, with fitness}
\end{equation}
where the fitness functions 
\begin{equation}
\frac{f_{\alpha}^{(t)}}{\sum_{\beta=1}^{m}f_{\beta}^{(t)}w_{\beta}^{(t)}}\label{eq: fitness functions}
\end{equation}
evolve too. This evolution naturally follows from the relation
(\ref{eq: variation Gibbs' weights}) linking Gibbs' weights to their
variations. 

Using the expression (\ref{eq: entropy variation and fluctuations})
for $\delta S$, we find that the stationary entropy condition $\delta S=0$
with respect to a variation $\delta\vec{f}$ can be equivalently expressed
by 
\begin{equation}
\left(\sum_{\alpha=1}^{m}e^{f_{\alpha}}f_{\alpha}\delta f_{\alpha}\right)\cdot\left(\sum_{\beta=1}^{m}e^{f_{\beta}}\right)=\left(\sum_{\alpha=1}^{m}e^{f_{\alpha}}f_{\alpha}\right)\cdot\left(\sum_{\beta=1}^{m}e^{f_{\beta}}\delta f_{\beta}\right).\label{eq: premise same expectation different distributions}
\end{equation}
The previous formula implies the equivalence of the expectation values
of $\delta f_{\alpha}$ with respect to the two distributions $\{w_{\alpha}(\boldsymbol{x})\}$
and $\{w_{\alpha}^{(1)}(\boldsymbol{x})\}$, namely 
\begin{equation}
\left\langle \delta f\right\rangle _{w_{1}}:=\sum_{\alpha=1}^{m}\frac{e^{f_{\alpha}}f_{\alpha}}{\sum_{\beta=1}^{m}e^{f_{\beta}}f_{\beta}}\delta f_{\alpha}=\sum_{\alpha=1}^{m}\frac{e^{f_{\alpha}}}{\sum_{\beta=1}^{m}e^{f_{\beta}}}\delta f_{\alpha}=:\left\langle \delta f\right\rangle _{w}.\label{eq: same expectation values}
\end{equation}
Similarly, from (\ref{eq: premise same expectation different distributions})
one can obtain a different set of fitness functions from the variations
$\delta f_{\alpha}$, that is 
\begin{equation}
\hat{w}_{\alpha}^{(1)}(\boldsymbol{x}):=\frac{e^{f_{\alpha}}\cdot\delta f_{\alpha}}{\sum_{\beta=1}^{m}e^{f_{\beta}}\delta f_{\beta}},\quad\alpha\in\{1,\dots,m\}\label{eq: discrete replicator dynamics, df, first step}
\end{equation}
and the corresponding formulation for the stationary entropy condition 
\begin{equation}
\langle f \rangle_{\hat{w}_{\alpha}^{1}}:=\sum_{\alpha=1}^{m}\hat{w}_{\alpha}^{(1)} \cdot f_{\alpha} =\sum_{\alpha=1}^{m} w_{\alpha}\cdot f_{\alpha}=:\langle f \rangle_{w_{\alpha}}.\label{eq: same expectation values, variations}
\end{equation}
In both these cases, we find that the entropy is stationary if and
only if the expectation value of $\delta f$ (respectively, $f$)
is the same for both the original Gibbs' distribution, and its evolution
(\ref{eq: discrete replicator dynamics, f, first step}) (respectively,
(\ref{eq: discrete replicator dynamics, df, first step})). 

Expectation values like those appearing in (\ref{eq: same expectation values, variations}) represent macroscopic observables. Different choices for the shift (\ref{eq: shift ground energy}) affect the constraints on them, which provide the fundamental information in the classical derivation of Gibbs' distribution through the maximum entropy principle. 
In this framework, the effects of this shift are manifest under the different actions on the sets $\{f_{\alpha}\}$ and $\{\delta f_{\alpha}\}$. Indeed, a shift for variations $\delta f_{\alpha}\mapsto \delta f_{\alpha}+M'$ preserves the equality between expectation values in (\ref{eq: same expectation values}), while a different choice $M'$ for $f_{\alpha}$ satisfying (\ref{eq: shift ground energy}) preserves $\left\langle \delta f\right\rangle _{w}$, but not necessarily $\left\langle \delta f\right\rangle _{w_{1}}$. This may result in additional requirements on the class of deformations $\delta f_{\alpha}$, together with the condition $\delta S=0$.  

\section{\label{sec: Ideal cases} Ideal and super-ideal cases}

General deformations can be divided into two classes. For the first
class, functions $f_{\alpha}$ remain unchanged and the deformations
$\delta f$ are due to the variations of the variables $x_{\alpha}\mapsto x_{\alpha}+\delta x_{\alpha}$.
For macroscopic systems, this is the change of the state of the system
due to the change of parameters which characterize it. For the second
class, the functions $f_{\alpha}$ are changed. This describes the
transformation of one statistical hypersurfaces to another one, or the transition
from one macroscopic system to another one (close to the original).
We will consider below both classes of deformations. 

The ideal case where $f_{\alpha}=b_{\alpha}+\sum_{i=1}^{n}a_{\alpha i}x_{i}$,
i.e. $\vec{f}=\vec{b}+\mathbf{A}\boldsymbol{x}$, is a special situation
where many geometric characteristics can be explicitly expressed and
related to physically relevant quantities. The first basic transformation
that one can consider is the translation $x_{i}\mapsto x_{i}+v_{i}\tau$.
Thus one gets 
\begin{align}
\delta f_{\alpha}= & \tau\cdot\sum_{i=1}^{n}a_{\alpha i}v_{i},\label{eq: ideal case, shift, variation functions}\\
\delta w_{\alpha}= & \tau w_{\alpha}\cdot\left(\sum_{i=1}^{n}a_{\alpha i}v_{i}-\sum_{\beta=1}^{m}\sum_{i=1}^{n}w_{\beta}a_{\beta i}v_{i}\right),\label{eq: ideal case, shift, variation weights}\\
\delta S= & \tau\cdot\left(\langle\vec{w}|\vec{b}\rangle+\vec{w}^{T}\mathbf{A}\boldsymbol{x}\right)\cdot(\vec{w}^{T}\mathbf{A}\boldsymbol{v})-\tau\cdot\left(\sum_{\alpha=1}^{m}w_{\alpha}b_{\alpha}\cdot(\mathbf{A}\boldsymbol{v})_{\alpha}+\sum_{\alpha=1}^{m}w_{\alpha}\cdot(\mathbf{A}\boldsymbol{x})_{\alpha}\cdot(\mathbf{A}\boldsymbol{v})_{\alpha}\right).\label{eq: ideal case, shift, variation entropy}
\end{align}
Having fixed $\tau$ and $\boldsymbol{v}$, the sign of $\delta S$
depends on the norm of $\boldsymbol{x}$. Indeed, fixing the unit
vector parallel to $\boldsymbol{x}$, i.e. $\vec{z}:=||\boldsymbol{x}||^{-1}\cdot\boldsymbol{x}$,
we see that 
\[
||\vec{x}||\cdot\left((\vec{w}^{T}\mathbf{A}\boldsymbol{z})\cdot(\vec{w}^{T}\mathbf{A}\boldsymbol{v})-\sum_{\alpha=1}^{m}w_{\alpha}\cdot(\mathbf{A}\boldsymbol{z})_{\alpha}\cdot(\mathbf{A}\boldsymbol{v})_{\alpha}\right)\geq-\langle\vec{w}|\vec{b}\rangle\cdot(\vec{w}^{T}\mathbf{A}\boldsymbol{v})+\sum_{\alpha=1}^{m}w_{\alpha}b_{\alpha}(\mathbf{A}\boldsymbol{v})_{\alpha}
\]
is satisfied on a half-line in $\mathbb{R}$. 

From the results in \cite{AK2016}, an ideal hypersurface has vanishing
Gauss-Kronecker curvature if and only if we can act on a certain nonvanishing vector in $\mathbb{R}^{n}$ through $\mathbf{A}$ to get a vector in $\mathbb{R}^{m}$ with all equal components, namely, $\vec{1}_{m}\in\mathrm{Im}(\mathbf{A})$
or $\vec{0}_{m}\in\mathrm{Im}(\mathbf{A})$. This property relates
to variations of ideal models preserving the statistical weights.
In fact the existence of such a vector $\boldsymbol{v}$ satisfying
\begin{equation}
\sum_{i=1}^{n}a_{\alpha i}\cdot v_{i}=\sum_{\beta=1}^{m}\sum_{i=1}^{n}w_{\beta}a_{\beta i}v_{i},\quad\alpha=1,\dots,m\label{eq: vanishing GK curvature ideal case}
\end{equation}
implies that the associated variation (\ref{eq: ideal case, shift, variation weights})
preserves all Gibbs' weights. Using the notation (\ref{eq: Hessian super-ideal skeleton}),
this means that the kernel of the matrix $\mathbf{H}\cdot\mathbf{A}$
is non-trivial. This happens when $\mathbf{A}$ has a non-trivial
kernel ($\vec{0}_{m}\in\mathrm{Im}(\mathbf{A})$), or when $\mathbf{A}(\mathbb{R}^{n})\cap\text{ker}(\mathbf{H})$
is non-trivial. As remarked in \cite{AK2016} (Proposition 3.1),
the vector $\vec{1}_{m}$ is the only eigenvector of $\mathbf{H}$
associated with the eigenvalue $0$, thus a direction $\boldsymbol{v}$ associated
to a variation of linear functions preserves all Gibbs' weights and,
hence, the entropy if and only if all the components of $\mathbf{A}\boldsymbol{v}$
are the same. This characterize a special class of isentropic variations
of the independent coordinates in ideals models, which corresponds to
a maximal permutational symmetry for the coordinates of the shift
vector $\mathbf{A}\boldsymbol{v}$. 

Among ideal models, the homogeneous cases ($b_{\alpha}=0$ for all
$\alpha$) are of particular physical interest. We discuss them in
the following. 

\subsection{\label{subsec: Linear cases} Linear cases }

The entropy (\ref{eq: entropy geometric formula}) has two contributions:
the first one $\sqrt{\det\mathbf{g}}\vec{X}\cdot\vec{N}$ has a pure
geometrical meaning, since it is the scalar product between $\vec{X}$
and the hypersurface vector field $\sqrt{\det\mathbf{g}}\cdot\vec{N}$
with a given orientation. The second contribution $\sum_{\alpha=1}^{m}w_{\alpha}\left(\sum_{i=1}^{n}x_{i}\frac{\partial f_{\alpha}}{\partial x_{i}}-f_{\alpha}\right)$
takes into account the deviation from the homogeneous \textit{linear
model} of degree $1$. Thus, it vanishes if all the functions $f_{\alpha}$
are linear, $f_{\alpha}(\boldsymbol{x})=\sum_{\alpha}a_{\alpha i}x_{i}$,
which gives the formula (\ref{eq: entropy geometric formula, super-ideal}). 

The distinction between these two contributions is physically relevant
since Gibbs' weights are often defined in terms of a scalar product
between a vector of control variables (temperature, chemical potentials)
and one of associated extensive quantities (energy levels, number
of molecules of a given species in the grand canonical partition function).
In the ideal case, the deviation from homogeneity is described by
the inclusion of constant terms $b_{\alpha}$ in the functions $f_{\alpha}$.
This results in a contribution $-\sum_{\alpha=1}^{m}w_{\alpha}b_{\alpha}$
in (\ref{eq: entropy geometric formula}). Depending on the physical
interpretation of the variables $\boldsymbol{x}$, this contribution
can be linked to a form of residual entropy, which is typical of highly
degenerate and frustrated systems \cite{Pauling1935,Lieb1967,Diep2005}. 

This geometric formulation allows one to go beyond the local study
of statistical characteristics. The global aspects have particular
significance in the framework of linear models, as can be seen looking
at the integral 
\begin{equation}
\int_{\Sigma}S(\boldsymbol{x})d\boldsymbol{x}\label{eq: Bayesian, induced prior}
\end{equation}
where $\Sigma$ is a compact subset of the statistical hypersurface
$V_{n}$. The focus on this quantity is suggested by the formula
(\ref{eq: entropy geometric formula}) for the entropy, since it corresponds
to the integrating function over a (hyper-)surface. In more detail,
we can consider two statistical hypersurfaces $V_{n,1}$ and
$V_{n,2}$ associated with two distinct sets of linear functions
$\{f_{1},\dots,f_{m}\}$ and $\{g_{1},\dots,g_{l}\}$. In order to
select a compact subset of $\mathbb{R}^{n+1}$, we fix any cone $\mathcal{C}\subseteq\mathbb{R}^{n+1}$
intersecting the two hypersurfaces and enclosing a compact region
$\Omega\subseteq\mathbb{R}^{n+1}$. An example of this construction
is presented in Fig. \ref{fig: section comparison linear hypersurfaces}.
\begin{figure}[tph]
\begin{centering}
\includegraphics[scale=0.5]{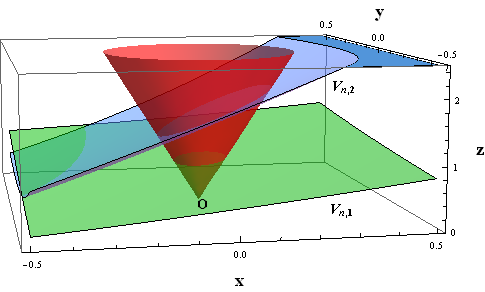}
\par\end{centering}
\centering{}\caption{\label{fig: section comparison linear hypersurfaces}A representation
of the comparison of two linear statistical hypersurfaces (in green
and in blue) with an intersecting cone (in red).}
\end{figure}
The vector $\vec{X}$ is orthogonal to the normal vector $\vec{N}$
for the cone, so 
\begin{equation}
\int_{\kappa}h\cdot(\vec{N}\cdot\vec{X})d\boldsymbol{x}=0\label{eq: vanishing contribution cone}
\end{equation}
for any compact subset $\kappa\subseteq\mathcal{C}$ and integrable
function $h$ with support $\kappa$. We can denote the boundaries $\Sigma_{1}:=V_{n,1}\cap\Omega$,
$\Sigma_{2}:=V_{n,2}\cap\Omega$, and $\Sigma_{0}:=\mathcal{C}\cap\Omega$
to get a region $\Omega\subseteq\mathbb{R}^{n+1}$ with piecewise
smooth boundary. 
The inclusion of a vanishing contribution (\ref{eq: vanishing contribution cone}) and a coherent choice for the orientations of the two hypersurfaces $V_{n,1}$ and $V_{n,2}$ allow us to invoke the Divergence theorem in $\mathbb{R}^{n+1}$ and get 
\begin{align}
\int_{\Sigma_{2}}S(\boldsymbol{x})d\boldsymbol{x}-\int_{\Sigma_{1}}S(\boldsymbol{x})d\boldsymbol{x} = & \int_{\Sigma_{2}}S(\boldsymbol{x})d\boldsymbol{x}+\int_{-\Sigma_{1}}S(\boldsymbol{x})d\boldsymbol{x} \nonumber \\
= & \int_{\Sigma_{2}}S(\boldsymbol{x})d\boldsymbol{x}+\int_{-\Sigma_{1}}S(\boldsymbol{x})d\boldsymbol{x}+\int_{\kappa}\sqrt{\det\mathbf{g}}\vec{X}\cdot\vec{N}d\boldsymbol{x} \nonumber \\
= & \int_{\Omega}\nabla\cdot\vec{X}d\boldsymbol{x}^{n+1} \nonumber \\
= & \int_{\Omega}\sum_{\alpha=1}^{n+1}\frac{\partial x_{\alpha}}{\partial x_{\alpha}}d\boldsymbol{x}^{n+1}\nonumber \\
= & \int_{\Omega}(n+1)\cdot d\boldsymbol{x}^{n+1} \nonumber \\
= & (n+1)\cdot\mathrm{volume}(\Omega).
\end{align}
This means, once a suitable cone is given, such a compact region and its volume are well-defined, and we can associate the variation of entropy between two linear models to the volume enclosed by them. 

It is worth stressing the physical meaning of this geometric construction,
with reference to the two geometric objects introduced in this section,
namely, the cone $\mathcal{C}$ and the integral characteristic (\ref{eq: Bayesian, induced prior}).
The cone $\mathcal{C}$ naturally arises when subsets of variables
$\boldsymbol{x}$ connected by a scaling factor $k$ are considered.
In more detail, in linear models the scaling $\boldsymbol{x}\mapsto(k_{B})^{-1}\cdot\boldsymbol{x}$
is related to the behaviour of the statistical hypersurface at large
values of the variables, the scaling is linked to a formal variation
for Boltzmann's constant $k_{B}$, and the limit $k_{B}\rightarrow0$
is associated with the \textit{tropical limit} (of the first kind)
of the statistical hypersurfaces \cite{AK2016} and the associated
statistical amoebas \cite{AK2018}. 

For what concerns the integration step, we first note that the canonical
volume form for this Riemannian manifold has been used. In turn, the
metric is inherited by the ambient space $\mathbb{R}^{n+1}$. Thus,
in this framework the natural emergence of an induced metric, its
relation with the entropy production, and the conditions needed to
apply Stokes' theorem, are all consequences of a single principle,
that is the embedding in an Euclidean space defined by the free energy
(\ref{eq: basic free energy formula}). 

On the other hand, different measures over $V_{n}$ may be considered. We have implicitly dealt with this aspect in the choice of an arbitrary compact subset $\Sigma$ of the statistical
hypersurface $\Sigma$. In general, the introduction of a measure
over the variable space associated with a probability distribution
(Gibbs' weights in our model) has a Bayesian interpretation \cite{Habeck2014}.
Indeed, the parameters $\boldsymbol{x}$ defining Gibbs' distribution
$\vec{w}$ are not fixed, and the space of parameters $\mathbb{R}^{n}$
becomes a measure space with a choice of the measure $d\nu$ on it.
This can be interpreted as the occurrence of parameters $\boldsymbol{x}$
with some weight expressed by the density $d\nu$, which is a typical
approach in Bayesian probability. The integral (\ref{eq: Bayesian, induced prior})
can then be seen as an ``average value'' for the entropy (\ref{eq: entropy geometric formula, super-ideal}) over $\Sigma$, with weight induced by the uniform distribution over $\mathbb{R}^{n+1}$, or as the average value of the scalar product $\vec{X}\cdot\vec{N}$ with measure $\sqrt{\det\mathbf{g}}$.
It is remarkable that a similar measure has been considered in information
theory, i.e. Jeffreys prior \cite{Jeffreys1946}, which is however
associated with a Riemannian metric defined from the entropy, namely,
the Fisher metric \cite{Amari2000}. 

The introduction of different measures in (\ref{eq: Bayesian, induced prior})
and their physical interpretation will be studied in a separate work. 

\subsection{\label{subsec: Super ideal cases} Super-ideal cases}

The super-ideal case where $m=n$ and $f_{\alpha}=x_{\alpha}$, $\alpha=1,\dots,n$,
is a fundamental one, since it provides the basic structure for our
model. On the other hand, it is also a special case, since a variation
of the variables implies a variation of the functions, so the two
kinds of deformations mentioned above coincide. We start from (\ref{eq: generalized entropy, geometric})
and consider variations of the type $x_{\alpha}\mapsto x_{\alpha}+v_{\alpha}\tau$
with $\tau\ll1$. In this way, the formula for the variation of entropy
is 
\begin{equation}
\delta S=\tau\cdot\left(\sum_{\alpha=1}^{n}w_{\alpha}x_{\alpha}\right)\cdot\left(\sum_{\beta=1}^{n}w_{\beta}v_{\beta}\right)-\tau\cdot\sum_{\alpha=1}^{n}w_{\alpha}x_{\alpha}v_{\alpha}.\label{eq: variation entropy super-ideal case}
\end{equation}
The choice $v_{\alpha}=v_{\beta}$ for all $\beta=1,\dots,n$ coincides
with the direction $(1,1,\dots,1)$ and satisfies the detailed balance
condition $v_{\alpha}=\sum_{\beta=1}^{n}w_{\beta}v_{\beta}$. In such
a case, one trivially has $\delta S=0$ since each term $\delta w_{\alpha}$
individually vanishes. 

The physical requirement on the entropy growth implies a natural direction
for the deformation parameter $\tau$, namely, its sign coincides
with the sign for $\sum_{\alpha=1}^{n}w_{\alpha}x_{\alpha}v_{\alpha}-\left(\sum_{\alpha=1}^{n}w_{\alpha}x_{\alpha}\right)\cdot\left(\sum_{\beta=1}^{n}w_{\beta}v_{\beta}\right).$
For a given choice of $\tau$, we can consider the set of vectors
$(v_{1},\dots,v_{m})$ such that $\delta S$ is non-negative. From
(\ref{eq: variation entropy super-ideal case}), we can see that,
for each $\boldsymbol{v}_{1},\boldsymbol{v}_{2}$ satisfying this
condition, their positive combinations $c_{1}\boldsymbol{v}_{1}+c_{2}\boldsymbol{v}_{2}$
satisfy it too, and $\boldsymbol{v}$ lies in this cone if and only
if $-\boldsymbol{v}$ does not. Thus, at each point $\boldsymbol{x}$,
it is defined a half-space $\mathcal{H}_{\boldsymbol{x}}\subseteq\mathbb{R}^{m}$
of directions of deformations satisfying the entropy increase principle.
From the previous observation, $(1,\dots,1)$ lies on the boundary
of such a half-space, as can be seen from $\delta S|_{\boldsymbol{v}=(1,\dots,1)}=0$
and the linearity of (\ref{eq: variation entropy super-ideal case})
with respect to $\boldsymbol{v}$. 

With reference to the integral characteristics of the kind (\ref{eq: Bayesian, induced prior}),
we can explicitly carry out its calculation for the super-ideal case
at $n=2$ with variables $(x_{1},x_{2})\equiv(x,y)$. At this purpose,
we choose a general rectangular domain $\Sigma:=[-c;c]\times[-c;c]$,
$c>0$, and get
\begin{align}
\overline{S}_{2}(c):=\int_{\Sigma}Sdxdy= & \int_{\Sigma}\ln(e^{x}+e^{y})dxdy-\int_{D}\frac{e^{x}x}{e^{x}+e^{y}}dxdy-\int_{D}\frac{e^{y}y}{e^{x}+e^{y}}dxdy\nonumber \\
= & 4c\text{Li}_{2}\left(-e^{2c}\right)-6\text{Li}_{3}\left(-e^{2c}\right)-\frac{2\pi^{2}c}{3}-\frac{9\zeta(3)}{2}\label{eq: integral entropy}
\end{align}
where $\mathrm{Li}_{2}$, $\mathrm{Li}_{3}$, and $\zeta$ denote the dilogarithm, trilogarithm, and zeta function, respectively. From the asymptotic behaviours at $c\rightarrow+\infty$
\begin{align}
\text{Li}_{2}(-e^{2c})\sim & -2\cdot c^{2}-\frac{\pi^{2}}{6}+\mathcal{O}(c^{-2})+\mathcal{O}(e^{-2c}),\label{eq: asymptotic dilogarithm}\\
\text{Li}_{3}(-e^{2c})\sim & -\frac{4\cdot c^{3}}{3}-\frac{\pi^{2}\cdot c}{3}+\mathcal{O}(c^{-2})+\mathcal{O}(e^{-2c})\label{eq: asymptotic trilogarithm}
\end{align}
the integral (\ref{eq: integral entropy}) increases as 
\begin{equation}
\overline{S}_{2}(c)\sim\frac{2\pi^{2}c}{3}-\frac{9\zeta(3)}{2},\quad c\rightarrow\infty.
\label{eq: asymptotic behaviour integral entropy}
\end{equation}
\begin{figure}[tph]
\centering{}\includegraphics[scale=0.55]{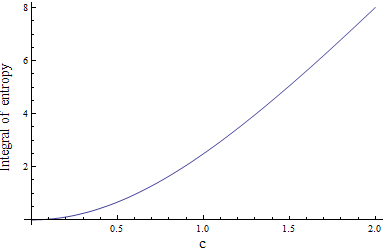}\hfill
\includegraphics[scale=0.55]{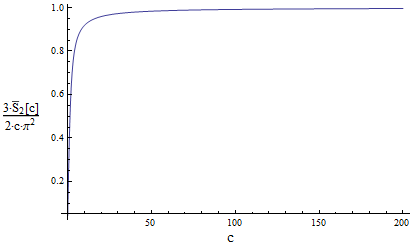}
\caption{Integral of entropy for super-ideal mapping at $n=2$ over rectangular
domain (a) and its ratio to the asymptotic behaviour described in (\ref{eq: asymptotic behaviour integral entropy}) (b).}
\end{figure}

\section{\label{sec: Conclusion} Conclusion }

In this paper, we have discussed different connections between physical
aspects of composite systems and their geometric descriptions. Main
attention has been paid to the variations of the functions $f_{\alpha}$
associated with a generalized free energy and their relation with
the second law of thermodynamics. Both deformations induced by a set
of control variables $\boldsymbol{x}\in\mathbb{R}^{n}$, which mimic
quasi-static variations within the same statistical system, and variations
of the functions $f_{\alpha}\mapsto f_{\alpha}+\delta f_{\alpha}$,
which change the statistical model, have been considered. The variation
of associated geometric quantities (Riemann curvature tensor, Gauss-Kronecker
curvature) have been provided, and a more detailed discussion on the
variation of entropy has been carried out. While general variations
have been considered, a characterization of the basic structure of
the fundamental statistical mapping (\ref{eq: basic free energy formula})
has been given. Ideal, especially linear and super-ideal cases, have
been discussed more extensively, also considering global characteristics
and their interpretation in a Bayesian framework. Finally, a formulation
of the stationary entropy condition in terms of discrete replicator
dynamics for Gibbs' weights has been discussed. 

This study indicates the way for possible further investigations,
which mainly concern the deviation from standard, equilibrium
statistical models. For instance, the exploration of nonlinearities
in functions $f_{\alpha}$, and their effects on Gauss-Kronecker curvature
and on principal curvatures can be explored, also considering some
specific physical systems that manifest non-homogeneity, small-size
effects, or non-equilibrium phenomena. Likewise, the graph-theoretic interpretation
of the dependence of statistical weights on functions $f_{\alpha}$ 
can be studied in more generality. This relates to the condition (\ref{eq: same expectation values})
leading to the discrete replicator equation. Thus, different weights
on the graph with Laplacian (\ref{eq: constrained Laplacian matrix for graph}),
and different evolutions (\ref{eq: discrete replicator dynamics})
may be useful to describe models out of the equilibrium. 

The maximum entropy condition has played a central role in the present
considerations. However, the relation with other geometric structures
that describe the statistical properties of the physical system have
to be investigated. In this direction, transformations acting on the
statistical amoebas \cite{AK2018} and the effect on nonlinear deformations
on the multiscale tropical limit \cite{AK2016} will be considered
too.

\section*{\label{sec: Appendix A} Appendix A: Proof of Proposition \ref{prop: fundamental statistical map from balancing assumption}}
\begin{proof}
From (\ref{eq: ansatz for metric, 2}), we get 
\begin{equation}
\frac{\partial\log h_{\alpha}}{\partial f_{\beta}}=\delta_{\alpha\beta}-h_{\beta}.\label{eq: log h derivative}
\end{equation}
Let us define 

\begin{equation}
l_{\alpha\beta}:=\log h_{\alpha}-\log h_{\beta},\quad\alpha,\beta=1,\dots,m.\label{eq: difference log prob}
\end{equation}
At $\alpha\neq\gamma\neq\beta$ one has ${\displaystyle \frac{\partial}{\partial f_{\gamma}}l_{\alpha\beta}=0}$,
so $l_{\alpha\beta}=l_{\alpha\beta}(f_{\alpha},f_{\beta})$ only depends
on $f_{\alpha}$ and $f_{\beta}$. Furthermore, at $\alpha=\beta$
one has $l_{\alpha\beta}=0$, so we focus on the case $\alpha\neq\beta$.
From (\ref{eq: log h derivative}) one easily gets 
\[
\frac{\partial l_{\alpha\beta}}{\partial f_{\alpha}}=1,\quad{\displaystyle \frac{\partial l_{\alpha\beta}}{\partial f_{\beta}}=}-1,
\]
so we find 
\begin{equation}
{\displaystyle {\displaystyle l_{\alpha\beta}=f_{\alpha}-f_{\beta}+c_{\alpha\beta}}}\label{eq: logarithm form, derivation}
\end{equation}
for some constants $c_{\alpha\beta}$. Clearly, $c_{\alpha\beta}=-c_{\beta\alpha}$
and 
\begin{align}
0= & l_{\alpha\beta}+l_{\beta\gamma}+l_{\gamma\alpha}\nonumber \\
= & c_{\alpha\beta}+c_{\beta\gamma}+c_{\gamma\alpha}\label{eq: closure c_ij}
\end{align}
for all pairwise distinct $\alpha,\beta,\gamma=1,\dots,m$. Let $\sigma_{\alpha}:=c_{1\alpha}$
for all $\alpha=1,\dots,m$, where $\sigma_{1}=c_{1,1}=0$. Then $c_{1\alpha}=-c_{\alpha1}=\sigma_{\alpha}-\sigma_{1}$
and $c_{\alpha\beta}=-c_{1\alpha}-c_{\beta1}=c_{1\beta}-c_{1\alpha}=\sigma_{\beta}-\sigma_{\alpha}$
at $\alpha\neq1\neq\beta$. So one has $c_{\alpha\beta}=\sigma_{\beta}-\sigma_{\alpha}$
for all $\alpha,\beta$. Thus (\ref{eq: logarithm form, derivation})
implies that ${\displaystyle \log h_{\alpha}-\log h_{\beta}=f_{\alpha}-f_{\beta}+\sigma_{\beta}-\sigma_{\alpha}}$.
Then the quantity 
\begin{equation}
H:=e^{\sigma_{\alpha}-f_{\alpha}}h_{\alpha}\label{eq: reciprocal of quasi-free energy}
\end{equation}
is independent on $\alpha$. In order to get the compatibility of
(\ref{eq: ansatz for metric, 2}) and (\ref{eq: reciprocal of quasi-free energy}),
one has 
\begin{align}
\frac{\partial h_{\alpha}}{\partial f_{\beta}}= & \delta_{\alpha\beta}H\cdot\exp(f_{\alpha}-\sigma_{\alpha})+\frac{\partial H}{\partial f_{\beta}}\exp(f_{\alpha}-\sigma_{\alpha})\nonumber \\
= & \delta_{\alpha\beta}H\cdot\exp(f_{\alpha}-\sigma_{\alpha})-H^{2}\cdot\exp(f_{\alpha}-\sigma_{\alpha}+f_{\beta}-\sigma_{\beta})\label{eq: condition of F, derivation}
\end{align}
which implies that 
\[
\frac{\partial}{\partial f_{\beta}}H^{-1}=\exp(f_{\beta}-\sigma_{\beta}).
\]
Hence one gets $H^{-1}=\exp(f_{1}-\sigma_{1})+H_{1}(f_{2},\dots,f_{m})^{-1}$
for a function $H_{1}$ of $f_{2},\dots,f_{m}$, which leads to 
\[
\frac{\partial}{\partial f_{2}}\left(\frac{1}{H}\right)=\frac{\partial}{\partial f_{2}}\left[\frac{1}{H_{1}(f_{2},\dots,f_{m})}\right]=\exp(f_{2}-\sigma_{2}),
\]
so $H^{-1}=\exp(f_{1}-\sigma_{1})+\exp(f_{2}-\sigma_{2})+H_{2}(f_{3},\dots,f_{m})^{-1}$.
Iterating this process, one finds $\tilde{F}:=H^{-1}=\gamma+\sum_{\alpha=1}^{m}\exp(f_{\alpha}-\sigma_{\alpha})$
where $\gamma$ is a constant, so 

\begin{equation}
{\displaystyle h_{\alpha}=\frac{\partial\tilde{F}}{\partial f_{\alpha}}=\frac{\exp(f_{\alpha}-\sigma_{\alpha})}{\gamma+\sum_{k=1}^{d}\exp(f_{\alpha}-\sigma_{\alpha})}},\quad\alpha=1,\dots,m.\label{eq: general explicit ansatz}
\end{equation}
\end{proof}

\end{document}